\documentclass[11pt]{article}
\usepackage[T1]{fontenc}
\usepackage[utf8]{inputenc}
\usepackage{lmodern}
\usepackage{geometry}
\usepackage{amsmath, amssymb}
\usepackage{titlesec}
\usepackage{setspace}
\usepackage{amsthm}
\usepackage{enumitem}
\usepackage{hyperref}

\titlespacing*{\section}{0pt}{2.5ex plus 0.8ex minus 0.3ex}{1.5ex plus 0.3ex}

\renewenvironment{thebibliography}[1]%
  {\fontsize{10}{11.5}\selectfont\begin{oldthebibliography}{#1}%
   \setlength{\itemsep}{0pt}%
   \setlength{\parskip}{0pt}%
   \setlength{\parsep}{0pt}}%
  {\end{oldthebibliography}}

\newtheorem{theorem}{Theorem}
\newtheorem{proposition}[theorem]{Proposition}
\newtheorem{lemma}[theorem]{Lemma}

\theoremstyle{definition}
\newtheorem{definition}[theorem]{Definition}
\theoremstyle{remark}
\newtheorem{remark}[theorem]{Remark}

\newcommand{\dbci}{DBCI}
\newcommand{\pfwd}{p_{\mathrm{fwd}}}
\newcommand{\pbwd}{p_{\mathrm{bwd}}}
\newcommand{\pdbci}{p_{\dbci}}

\title{Dual Boundary Condition Inference:\\
       A Two-Boundary Product Rule and Its Implications}
\author{%
  Luis Razo\thanks{European Institute of Science in Management (EISM), Barcelona, Spain}%
  \and
  Eliahu Cohen\thanks{Faculty of Engineering and the Institute of Nanotechnology and Advanced Materials, Bar-Ilan University, Israel}%
}
\date{}

\begin{document}
\maketitle

\emergencystretch=2em
\begin{abstract}
Many inference problems are constrained from two sides: forward information pointing toward an answer, and a second constraint on the outcomes. The rule studied here, \emph{Dual Boundary Condition Inference} (\dbci{}), multiplies them outcome by outcome and renormalises. The rule is not new: it is the product-of-experts form and the unit-exponent member of the logarithmic-pooling family. The paper contributes a classification and reducibility characterisation. The algebra does not fix the exponent at which the inputs combine: given the statistic $\log\pfwd+\log\pbwd$, which already builds in the product form, maximum entropy supplies the family $(\pfwd\pbwd)^{\lambda}$ and selects no member. If the inputs are read as two equally weighted opinions, three requirements select $\lambda=\tfrac12$: unanimity preservation, consistency under a common Bayesian update, and minimal symmetric Kullback--Leibler divergence to the inputs. If they are read as separately applied factors, two requirements select $\lambda=1$: a neutral second input must leave the first unchanged, and an update to one input must pass through unchanged. For a rank-one intermediate projective measurement, the Aharonov--Bergmann--Lebowitz (ABL) rule realises the unit-exponent factor form without an additional exponent parameter. Each reading is thus supported in the same three ways: by what it means for an input to add nothing, by how new information enters, and by an independent derivation. Every positive exponent orders outcomes identically, so the choice is invisible to picking the most probable outcome, but not to scoring a class by its total mass. Under the relevant identifications, the same product form gives Bayes' rule, while a final effect proportional to the identity returns the ordinary forward Born probabilities. The second result is a \emph{Two-Boundary Reducibility Criterion}: \dbci{} factors through the forward boundary precisely when the effective backward boundary does, after weights differing only by a positive factor are identified. Thus a fixed prior or model-fixed structural constraint adds no distinctions between cases sharing the same forward boundary, though it may encode substantial information and still materially affect the result. The resulting rule is substrate-independent.
\end{abstract}

\section{Introduction}
\label{sec:introduction}

Most inference problems have one boundary. Bayes' rule combines a likelihood with a prior; maximum-likelihood estimation maximises a single density; the Born rule maps a pre-selected quantum state to outcome probabilities. In each case, once the model specification is fixed, a single case-varying input determines the distribution over outcomes, and the rule for using it is unambiguous. A prior is not thereby empty of information; it simply does not vary with the case, a distinction made precise in \S\ref{sec:scope}.

Other problems are two-sided in form: two separately specified constraints bear on the same outcome space, and the task is to combine them. (Whether the second constraint supplies information the first does not is a further question, addressed in \S\ref{sec:scope}; the examples below are two-sided in specification, not necessarily in informational content.) In post-selected quantum mechanics, the system is conditioned both on a pre-selected initial state and on a post-selected final measurement outcome; the rule for the probability of an intermediate measurement is the ABL rule of Aharonov, Bergmann, and Lebowitz~\cite{AharonovBergmannLebowitz1964}. In stabiliser quantum error correction, candidate error patterns are constrained both by an observed syndrome (forward evidence) and by the requirement that surviving configurations lie in the code subspace (backward structural constraint). In generative modelling and posterior inference, the latent-variable distribution is shaped simultaneously by data and by a fixed prior or regulariser. Each of these is a setting in which two probability distributions over a shared outcome space must be combined into one.

The rule for that combination is the multiplicative form
\begin{equation*}
\pdbci(x) \;=\; \frac{\pfwd(x)\,\pbwd(x)}{Z}, \qquad Z \;=\; \sum_{x'} \pfwd(x')\,\pbwd(x'),
\end{equation*}
which we call the \emph{Dual Boundary Condition Inference} (\dbci{}) operator. The thesis of the paper concerns the exponent. Given the statistic $\log\pfwd+\log\pbwd$, which already builds in the product form, maximum entropy supplies a one-parameter family of combinations, $(\pfwd\,\pbwd)^{\lambda}$, and does not select a member; what fixes $\lambda$ is a reading of what the two inputs are, and the three arguments below stand in the following relation to that question:

\begin{itemize}[noitemsep,topsep=2pt]
  \item \textbf{Quantum-mechanical.} For a rank-one intermediate projective measurement the ABL rule~\cite{AharonovBergmannLebowitz1964} is already this product, written out in that form in~\cite{AharonovCohenGrussLandsberger2014}. What the paper takes from it is not the identity but the selection: it fixes the exponent at one, with no further parameter once the inputs are identified as Born weights (\S\ref{sec:abl}, \S\ref{sec:unification}).
  \item \textbf{Information-theoretic.} Among all distributions on $X$ satisfying normalisation and a single moment constraint on the joint log-overlap $\log\pfwd + \log\pbwd$, the maximum-entropy solutions form the one-parameter family $(\pfwd\,\pbwd)^{\lambda}/Z_{\lambda}$. This argument supplies the family and selects no member of it: $\lambda$ is fixed only once the moment value is prescribed, which the argument does not supply (\S\ref{sec:maxent}).
  \item \textbf{Decision-theoretic.} The distribution closest in symmetric Kullback--Leibler divergence to both boundaries is the normalised geometric mean $\sqrt{\pfwd\pbwd}/Z'$. That is the $\lambda = \tfrac{1}{2}$ member, so this argument selects a different one, and the two selections disagree along the line \S\ref{sec:unification} makes precise (\S\ref{sec:kl}, \S\ref{sec:unification}).
\end{itemize}

The three are not three independent claims of correctness, nor three routes to one destination. One supplies the family; the other two select different members of it, and disagree. What that disagreement is, and why it is invisible to the comparison ordinarily made, is made precise in \S\ref{sec:unification}.

Beyond the one-shot combination, the operator carries structural properties under iteration. Fixing the backward boundary and applying it repeatedly to an evolving forward boundary defines a map $T$ on the simplex of distributions on $X$ (\S\ref{sec:properties}). Two iteration-level properties follow: for strictly positive $\pbwd$, the fixed points of $T$ are exactly the distributions on whose support $\pbwd$ is constant (\S\ref{sec:fixed-point-char}); and iteration concentrates mass on configurations of equal, minimum-available backward energy (\S\ref{sec:fixed-point}).

Familiar single-boundary rules are recovered as limits. When the backward boundary is uniform, the operator reduces to the forward boundary alone: the Born rule if the forward boundary is a squared amplitude, the likelihood rule if it is a data likelihood. Reading the backward boundary as a prior turns the operator into Bayes' rule, and taking logarithms turns it into a free energy whose total is the sum of a forward and a backward energy term. These limits, together with the operator's relationship to the two-state vector formalism~\cite{CohenAharonov2017,AharonovCohenLandsberger2017} and to weak measurement, are developed in \S\ref{sec:discussion}.

A word on where this form comes from. What the product algebra alone decides is neither whether the second factor varies in a way the first does not already determine, nor the exponent at which the two combine. Post-selected quantum mechanics bears on both, in different ways. A final condition can be specified independently of the initial one without redundancy or contradiction, even when the forward state is completely specified, subject only to the post-selection having nonzero probability, the condition $Z > 0$ appearing in the operator itself; under fixed deterministic dynamics there is no such freedom, and a final condition is either redundant or inconsistent~\cite{AharonovCohenGrussLandsberger2014}. The contrast is with determinism rather than with classicality, since classical smoothing~\cite{RauchTungStriebel1965} and Schr\"{o}dinger bridge~\cite{Leonard2014} problems also supply independent second boundaries; \S\ref{sec:scope} gives the criterion that decides which case obtains in a given family.

The exponent stands differently. It is not left open once the inputs have been named: \S\ref{sec:unification} shows that naming them is what closes it, since taking the two weights as separately applied factors and taking them as two equally weighted opinions force different members by elementary requirements that mention no quantum mechanics (Remark~\ref{rem:product-rules}). What is left open is which reading obtains, and that is a question about the setting rather than about the rule. Post-selection supplies a physical setting in which the factor reading is natural rather than stipulated: the two boundaries are fixed at distinct times by distinct procedures, neither determining the other. That places the setting on the factor side of the distinction; it does not show that quantum mechanics forces the reading. The claim is the distinctness of the two specifications, not a symmetry between them: a post-selected effect is not prepared in the sense a pre-selected state is, and the argument does not require that it be. The rank-one ABL reduction then returns unit exponents with no free parameter, once the two inputs are identified as Born weights of that pair. Any application must supply that identification, but it is made before any combination rule is in view. That agreement, between a physical derivation and requirements stated without reference to physics, and not the multiplication, is what this paper takes from the setting.

What we take is therefore the structure and not the physics. Once the operator is stated the quantum machinery can be discarded, and what remains is available wherever two boundary weights are separately specified on a shared outcome space. Whether that extraction is worth making is not settled by the algebra, which is standard. Feynman's observation applies directly: when two formulations yield exactly the same consequences there is no scientific ground for preferring either, and the reason to hold both is that each suggests different ideas about what to change next~\cite[ch.~7]{Feynman1965}. The two-boundary reading claims no consequence the normalised product does not already have. What it offers is a different set of questions, and it is answerable for whether those questions lead anywhere.

The contribution is a classification and a test. The classification: the exponent at which two boundary weights combine is not fixed by the algebra but by what the inputs are taken to be --- separately applied factors, or two equally weighted opinions about the same outcome --- and each reading is supported in the same three ways: an identity property, a covariance property and an independent derivation. The test: the Two-Boundary Reducibility Criterion converts ``second boundary'' from a description of how a weight was specified into a condition on whether it supplies case-specific information the forward boundary does not already determine.

This paper provides the formal framework; whether the questions it suggests lead anywhere in a particular substrate is an empirical matter. The remainder is organised as follows: \S\ref{sec:operator}--\S\ref{sec:kl} define the operator and give the three arguments; \S\ref{sec:unification} sets out their relation, and the disagreement between two of them; \S\ref{sec:properties} establishes the iteration properties; \S\ref{sec:scope} establishes the Two-Boundary Reducibility Criterion and derives from it a classification of backward boundaries by informational status relative to the case family; and \S\ref{sec:discussion} treats the single-boundary limits, the two-state-vector connection, and the operator's scope.

\section{The \dbci{} Operator}
\label{sec:operator}

Throughout the paper, $X$ denotes a finite outcome space and $\pfwd,\pbwd$ denote two separately specified boundary weights on $X$. The forward boundary carries evidence, such as a likelihood, generative-model prediction, or pre-selected boundary; the backward boundary carries structural, terminal, or post-selected information, such as a code-subspace weighting, post-selection effect, or prior-like structural constraint. The labels are conventional in the following sense: the combination law is symmetric under exchange of its arguments, so which weight is called \emph{forward} versus \emph{backward} reflects substrate-specific interpretation rather than any asymmetry in the rule. The roles the two boundaries play in the surrounding inference problem need not be symmetric: a substrate may update one while holding the other fixed, and \S\ref{sec:unification} turns on that distinction. For the comparison developed below, we restrict attention to the equal-exponent family $\mathcal{C}[(\pfwd\pbwd)^{\lambda}]$, where $\mathcal{C}$ denotes normalisation. The more general log-linear family $\mathcal{C}[\pfwd^{\lambda_1}\pbwd^{\lambda_2}]$ separates relative weighting from overall sharpness. Setting $\lambda_1=\lambda_2$ gives the two inputs equal log-weight; this is the exchange-symmetric subfamily studied here, not a consequence of calibration. Within this subfamily, \S\ref{sec:unification} compares the common exponents selected by opinion-pooling and unit-strength factor readings. The interpretive claims below concern $\lambda>0$; $\lambda=0$ discards both boundaries and $\lambda<0$ reverses their ordering.

\begin{definition}[\dbci{} Operator]
\label{def:operator}
Let $X$ be a finite outcome space and let $\pfwd,\pbwd\colon X\to[0,\infty)$ be nonnegative boundary weights with overlapping support, so that
\[
Z \;=\; \sum_{x\in X}\pfwd(x)\pbwd(x) \;>\; 0.
\]
The \dbci{} operator combines them into the probability distribution
\begin{equation}
\label{eq:operator-def}
\pdbci(x) \;=\; \frac{\pfwd(x)\,\pbwd(x)}{Z}.
\end{equation}
\end{definition}

\begin{remark}[Working regime]
\label{rem:well-defined}
The operator itself only requires $Z>0$, equivalently
\[
\operatorname{supp}(\pfwd)\cap \operatorname{supp}(\pbwd)\neq\emptyset.
\]
When logarithms, Kullback--Leibler divergences, or exponential-family parameterisations are used below, we additionally assume strict positivity on the relevant outcome space. Under that stricter assumption, $\log\pfwd$ and $\log\pbwd$ are finite everywhere, the symmetric-KL objective of \S\ref{sec:kl} is finite, and the backward boundary admits the energy parameterisation used in \S\ref{sec:properties}.
\end{remark}

\begin{remark}[Rescaling invariance]
\label{rem:rescaling}
The operator depends only on the product $\pfwd(x)\pbwd(x)$ up to positive scalar rescaling: replacing either boundary weight by a positive constant multiple leaves $\pdbci$ unchanged after normalisation. Thus $\pfwd$ and $\pbwd$ may be normalised probability distributions, likelihood-like factors, effects, priors, structural weights, or other positive log-potentials. This will be used implicitly in \S\ref{sec:discussion} when relating the operator to Bayes' rule, where the likelihood is naturally unnormalised.
\end{remark}

\begin{remark}[Neutral boundary]
\label{rem:neutral}
A uniform backward boundary represents the absence of a second constraint rather than a constraint that happens to be flat. Combining a forward boundary with it should therefore return that boundary unchanged, up to normalisation: nothing has been imposed, so nothing should move. The operator satisfies this, since $\pbwd$ constant gives $\pdbci = \mathcal{C}[\pfwd]$, the forward boundary up to the normalisation the operator always applies. The requirement is recorded here because it is a statement about what a boundary means, not about the form of the rule, and \S\ref{sec:unification} will ask which rules respect it.
\end{remark}

\begin{remark}[Scope of the operator]
\label{rem:operator-scope}
The operator applies when two boundary weights are separately specified on a shared outcome space and the task is to combine them into a posterior score or MAP decision rule. It is not a substitute for conditioning in a known joint model. Under conditional independence,
\[
P(x\mid y_1,y_2)\propto P(x)P(y_1\mid x)P(y_2\mid x),
\]
so multiplication is justified only after the inputs have been identified as non-duplicative factors of that posterior. Multiplying two posteriors generally counts their shared prior twice and requires the corresponding prior correction. When the operator represents the posterior, decisions based on it must account for the task's loss function, with MAP corresponding to zero-one loss on individual outcomes.
\end{remark}

\begin{remark}[Relation to established product rules]
\label{rem:product-rules}
The normalised product is presented in products of
experts~\cite{Hinton2002}, and it belongs to the weighted family
$\mathcal{C}\bigl[\prod_i p_i^{w_i}\bigr]$ of logarithmic
pooling, as the member with unit exponents. The two
should not be run together: logarithmic pooling is conventionally
normalised with $\sum_i w_i = 1$~\cite{Heskes1998}, which for two equally treated inputs
gives $\tfrac{1}{2}$ each and the geometric mean rather than the bare
product, and \S\ref{sec:unification} turns on exactly that difference. Closer still is the
Bayes-space geometry of van den Boogaart, Egozcue and
Pawlowsky-Glahn~\cite{vandenBoogaart2010}, whose objects are densities
identified up to positive rescaling and whose commutative group
operation, perturbation, is their pointwise product, with the
Radon--Nikodym derivative as its inverse. The operation is
Aitchison's~\cite{Aitchison1986}, who observed that perturbation on the
simplex is already familiar elsewhere in statistics as the operation by
which Bayes' formula carries a prior into a posterior. The operator of this paper is
that operation on a finite outcome space, and the quotient introduced
in \S\ref{sec:scope} is that identification, taken relative to the
support of the forward boundary. The contribution claimed here is
therefore not a new pooling algebra and not the proportionality
quotient, but the two-boundary interpretation, the comparison of the
quantum, maximum-entropy and decision-theoretic arguments bearing on the exponent, and
the criterion of \S\ref{sec:scope} for when the effective second argument adds no distinctions between cases sharing the same forward boundary.

The same product also appears wherever a forward and a backward message
are combined: in hidden Markov smoothing, and in the control- and
planning-as-inference literature~\cite{Attias2003,Levine2018}, where a
backward message carrying future reward or goal attainment multiplies a
forward term (the forward message in smoothing, a prior over actions or trajectories in planning) and is renormalised. There the product is derived from a
specified joint model, as a theorem about that model's posterior. That is not the only way the form can be earned: the pooling
literature discussed in \S\ref{sec:unification} characterises
multiplicative and geometric combination by axioms on the
aggregation rule itself, with no joint model in view.
Three warrants should therefore be kept apart: a
product may be derived from a joint model that is given; derived from
requirements imposed on the combination itself, no joint model being
available; or simply adopted, with neither supplied. The operator of
this paper has the second. \S\ref{sec:unification} shows that the
identity and covariance requirements belonging to a factor reading
select the unit exponent, and \S\ref{sec:abl} shows that the rank-one
ABL reduction returns the same member independently. What the operator
must not do is displace a joint model that is available and gives
something else (Remark~\ref{rem:operator-scope}).

An interpretive distinction is important, and it is the axis on which the
rest of the paper turns. If $\pfwd$ and $\pbwd$ are likelihood-like
factors, unit exponents give the usual product of factors. If they are
instead treated as two equally weighted probability opinions, the
conventional logarithmic pool uses exponents $1/2$ and $1/2$, yielding
the geometric mean. Symmetry alone therefore requires equal exponents
but does not fix their common value: what fixes it is a reading of what
the two inputs are. In the present framework, the unit exponents are
fixed by the rank-one ABL reduction once the inputs are identified as Born weights, or selected by the requirements
belonging to the reading of the inputs as unit-strength likelihood
factors. Both readings are represented below (\S\ref{sec:abl} takes the factor reading and \S\ref{sec:kl} the opinion reading, and \S\ref{sec:unification} sets out the resulting disagreement), so the exponent should be read throughout as a modelling choice rather than a fixed feature of the operator. One consequence should be stated here rather than left for a reader to discover: taken as an equal-source opinion pool, the unit-exponent operator is neither unanimity-preserving nor externally Bayesian under a common likelihood update, both of which select the half-exponent member instead. \S\ref{sec:unification} gives the computations and the reason those properties belong to the opinion reading rather than to the operator.
\end{remark}

\section{The Quantum Selection}
\label{sec:abl}

The first route to $\pdbci$ is direct: for a rank-one intermediate
projective measurement, the Aharonov--Bergmann--Lebowitz rule reduces
exactly to the normalised product of the forward and backward Born
weights. The identity is not new. It is an immediate consequence of
the ABL rule~\cite{AharonovBergmannLebowitz1964} once the intermediate
projectors are taken to be rank-one, and it is written in precisely
this form, as a normalised product of the two Born weights, in the
subsequent two-state-vector
literature~\cite{AharonovCohenGrussLandsberger2014}. It is restated
here because it is the premise the rest of the paper detaches from its
quantum setting, and because the load-bearing restriction is the
rank-one character of the measurement rather than any assumption about
phases in the boundary states.

For an intermediate measurement with rank-one projectors
$\hat P_n=|n\rangle\langle n|$, the ABL rule is
\begin{equation}
\label{eq:abl}
P_{\mathrm{ABL}}(n)
=
\frac{|\langle\phi|\hat P_n|\psi\rangle|^2}
{\sum_j|\langle\phi|\hat P_j|\psi\rangle|^2}.
\end{equation}
Write arbitrary normalised boundary states as
\begin{equation}
\label{eq:diagonal-states}
|\psi\rangle=\sum_x e^{i\alpha_x}\sqrt{\pfwd(x)}\,|x\rangle,
\qquad
|\phi\rangle=\sum_x e^{i\beta_x}\sqrt{\pbwd(x)}\,|x\rangle,
\end{equation}
where $\pfwd(x)=|\langle x|\psi\rangle|^2$ and
$\pbwd(x)=|\langle x|\phi\rangle|^2$.

\begin{proposition}[ABL Rank-One Reduction]
\label{thm:abl-reduction}
For arbitrary pure boundary states of the form
(\ref{eq:diagonal-states}) and rank-one projectors
$\hat P_n=|n\rangle\langle n|$,
\begin{equation}
\label{eq:dbci-from-abl}
P_{\mathrm{ABL}}(n)
=
\frac{\pfwd(n)\pbwd(n)}
{\sum_j\pfwd(j)\pbwd(j)}
=
\pdbci(n).
\end{equation}
\end{proposition}

\begin{proof}
For each outcome,
\[
\langle\phi|\hat P_n|\psi\rangle
=
e^{i(\alpha_n-\beta_n)}
\sqrt{\pfwd(n)\pbwd(n)}.
\]
The phase therefore cancels upon taking the modulus squared, giving
$|\langle\phi|\hat P_n|\psi\rangle|^2
=\pfwd(n)\pbwd(n)$. Substitution into the ABL rule proves the claim.
\end{proof}

The same statement holds more generally for an arbitrary density operator $\rho$ and a positive post-selection effect $E$, both propagated to the intermediate time. For the L\"uders instrument associated with the rank-one projectors, define
\[
\pfwd(x)=\langle x|\rho|x\rangle,
\qquad
\pbwd(x)=\langle x|E|x\rangle,
\]
where $\pbwd$ need not be normalised. Rank-one projectors then give
\[
\operatorname{Tr}(E P_x\rho P_x)
=
\langle x|\rho|x\rangle\langle x|E|x\rangle,
\]
and hence the same normalised product rule, provided the post-selection has non-zero probability. This is the rank-one
specialisation of the past-quantum-state construction
\cite{Gammelmark2013}.

\begin{remark}[Where the reduction fails]
For degenerate or coarse-grained projectors, coherent cross-terms
within a projected subspace may survive, and the ABL rule need not
reduce to a product of outcome probabilities.
\end{remark}

\begin{remark}[A case where the restriction bites]
Stabiliser syndrome measurement is degenerate by construction: a
syndrome identifies a coset of errors rather than an error, and the
projector onto a syndrome has rank $2^{k}$ for a code encoding $k$
logical qubits. The rank-one reduction above therefore does not apply to it, and not only on account of coherence: the trace over a
syndrome subspace need not factor into forward and backward marginals
even when the state carries no coherence within that subspace. Particular
degenerate cases can still factor after normalisation (an effect
proportional to the identity on each subspace contributes only the common
dimension factor, which the operator discards), so what degeneracy
removes is the general identity, not every instance of it. Two cases
should then be separated. Where the error model is explicitly stochastic,
the outcome space may be refined to classical error alternatives, and if
the two inputs are separately supplied factors on that refined space \S\ref{sec:maxent} supplies the family and the requirements of \S\ref{sec:unification} select the unit-exponent product within it, without reference to quantum mechanics;
the rank-one reduction is simply not the route to it. Where the errors within a
coset are coherent, neither the refinement nor the rank-one reduction
applies: enumeration presupposes that the alternatives are classical,
which is what is at issue.
\end{remark}

\section{The Maximum-Entropy Family}
\label{sec:maxent}

The maximum-entropy calculation~\cite{Jaynes1957} yields a variational characterisation of the equal-exponent family, conditional on choosing the log-statistic
\begin{equation}
\label{eq:gdef}
g(x) \;=\; \log\pfwd(x) + \log\pbwd(x).
\end{equation}
This choice imposes log-additivity and equal coefficients for the two boundaries. Maximising entropy at a prescribed value of $\langle g\rangle_W=\sum_x W(x)g(x)$ then selects a member of the family. It neither determines the statistic $g$ nor supplies the prescribed moment and hence does not select $\lambda$.

\begin{theorem}[Maximum-Entropy Family of Two-Boundary Combinations]
\label{thm:maxent}
Among all probability distributions $W$ on $X$ satisfying
\begin{enumerate}[label=\textup{(\roman*)}]
  \item normalisation $\sum_x W(x) = 1$, and
  \item the moment constraint $\sum_x W(x)\, g(x) = m$,
\end{enumerate}
the unique distribution maximising the Shannon entropy~\cite{Shannon1948} $H(W) = -\sum_x W(x)\log W(x)$ takes, for every $m$ in the interior of the convex hull of $\{g(x) : x \in X\}$, the form
\begin{equation}
\label{eq:maxent-form}
W^{*}_{\lambda}(x) \;=\; \frac{\bigl(\pfwd(x)\,\pbwd(x)\bigr)^{\lambda}}{\sum_{x'}\bigl(\pfwd(x')\,\pbwd(x')\bigr)^{\lambda}},
\end{equation}
where, for non-constant $g$, the Lagrange multiplier $\lambda$ is uniquely determined by $m$ and ranges over all of $\mathbb{R}$ as $m$ ranges over that interior. At either extreme feasible value of $m$ no finite $\lambda$ meets the constraint, and the entropy maximiser is instead the uniform distribution on the corresponding extremal level set $\{x \in X : g(x) = m\}$, which is the $\lambda \to +\infty$ limit of $W^{*}_{\lambda}$ at $m = \max_x g(x)$ and the $\lambda \to -\infty$ limit at $m = \min_x g(x)$. The one-parameter family $\{W^{*}_{\lambda} : \lambda \in \mathbb{R}\}$ therefore traces out all interior maximum-entropy solutions, with the two boundary cases as its limits. In the degenerate case where $g$ is constant on $X$ these statements do not apply: the convex hull is a single point, that value of $m$ is forced, every $\lambda$ satisfies the constraint, and the maximiser is the uniform distribution on $X$ for all of them, so $\lambda$ is non-identifiable. In every case the $\lambda = 1$ member is
\begin{equation}
W^{*}_{1}(x) \;=\; \frac{\pfwd(x)\,\pbwd(x)}{Z},
\qquad Z \;=\; \sum_{x'} \pfwd(x')\,\pbwd(x'),
\end{equation}
the operator $\pdbci$. The selection of $\lambda = 1$ over other members of the family is interpretive (Remark~\ref{rem:maxent-interp}): every member treats the two boundaries alike, since both carry the same exponent, and within this argument $\lambda = 1$ is distinguished only as the value at which each enters at unit exponent, neither amplified nor damped. It is not a value singled out by the maximum-entropy argument itself; \S\ref{sec:unification} gives requirements that do select it.
\end{theorem}

\emph{Proof in Appendix~\ref{app:proofs}.}

\begin{remark}[Informational interpretation]
\label{rem:maxent-interp}
The parameter $\lambda$ governs the strength with which $W$ is conditioned on the two boundaries. At $\lambda = 0$ the multiplier vanishes, the required moment is the one the uniform distribution already supplies, and $W^{*}$ is uniform; at $\lambda = 1$ each boundary enters at unit exponent and $W^{*} = \pdbci$; at $\lambda > 1$ the combination concentrates further on outcomes of larger joint product weight.
\end{remark}

\section{The Decision-Theoretic Selection}
\label{sec:kl}

The third argument is not a route to $\pdbci$: it selects a different member of the family of \S\ref{sec:maxent}. Throughout this section both boundaries are taken to be strictly positive probability distributions on $X$; this ensures that the divergences below are finite for every candidate $p'$, and by Remark~\ref{rem:rescaling} normalising them changes neither $\pdbci$ nor any conclusion drawn below. Given the two boundaries, what distribution $p'$ on $X$ is closest in Kullback--Leibler divergence~\cite{CoverThomas} to \emph{both} simultaneously? A natural answer is the minimiser of the boundary-symmetric divergence
\begin{equation}
\label{eq:kl-sym}
D_{\mathrm{sym}}(p') \;=\; \tfrac{1}{2}\!\left[D_{\mathrm{KL}}(p'\,\|\,\pfwd) + D_{\mathrm{KL}}(p'\,\|\,\pbwd)\right] \;=\; \tfrac{1}{2}\sum_{x} p'(x)\!\left[\ln\!\frac{p'(x)}{\pfwd(x)} + \ln\!\frac{p'(x)}{\pbwd(x)}\right].
\end{equation}
Here $D_{\mathrm{sym}}$ is the \emph{boundary-symmetric} objective: the average of the two forward divergences $D_{\mathrm{KL}}(p'\,\|\,\cdot)$ from $p'$ to each boundary, with the candidate in the first argument; it is not the Jeffreys symmetrisation of a single divergence between $\pfwd$ and $\pbwd$.

The answer is known. Minimising a weighted sum $\sum_i w_i D_{\mathrm{KL}}(p' \,\|\, q_i)$ over the first argument returns the normalised weighted geometric mean of the $q_i$; minimising over the second returns their arithmetic mean, the linear pool. The construction appears as the Kullback--Leibler average in distributed fusion~\cite{BattistelliChisci2014}, which gives logarithmic pooling a variational characterisation, and, for unnormalised weights, as a sided centroid in the Bregman-centroid literature~\cite{NielsenNock2009}. Lemma~\ref{lem:kl-min} is the equal-weight two-boundary case, restated here because the member it selects matters for \S\ref{sec:unification}. The naming conventions for which side is which are not stable across those literatures, so we state the minimisation explicitly rather than rely on them.

The minimiser is not $\pdbci$. It is $W^{*}_{1/2}$, the half-exponent member of the family of \S\ref{sec:maxent}, and it shares $\pdbci$'s MAP estimate for a reason that has nothing to do with either argument: every positive-exponent member of that family orders outcomes identically.

\begin{lemma}[Symmetric-KL Minimiser]
\label{lem:kl-min}
The unique distribution $p'$ on $X$ minimising (\ref{eq:kl-sym}) subject to $\sum_{x} p'(x) = 1$ is the normalised geometric mean
\begin{equation}
\label{eq:kl-min}
p'_{\mathrm{geo}}(x) \;=\; \frac{\sqrt{\pfwd(x)\,\pbwd(x)}}{\sum_{y}\sqrt{\pfwd(y)\,\pbwd(y)}}.
\end{equation}
\end{lemma}

\emph{Proof in Appendix~\ref{app:proofs}.}

\begin{proposition}[Ordering invariance across the family]
\label{thm:kl-opt}
For every $\lambda > 0$ the members $W^{*}_{\lambda} \propto (\pfwd\,\pbwd)^{\lambda}$ order the outcomes of $X$ identically, and therefore share a common maximiser set:
\begin{equation*}
\arg\max_{x \in X} W^{*}_{\lambda}(x) \;=\; \arg\max_{x \in X} \bigl(\pfwd(x)\,\pbwd(x)\bigr) \qquad \text{for all } \lambda > 0 .
\end{equation*}
In particular $\pdbci = W^{*}_{1}$ and $p'_{\mathrm{geo}} = W^{*}_{1/2}$ share a maximiser set, which is a property of the family rather than a relation between the two arguments that select them.
\end{proposition}

\emph{Proof in Appendix~\ref{app:proofs}.}

\begin{remark}[Distribution versus point estimate]
\label{rem:kl-decision}
The geometric-mean distribution $p'_{\mathrm{geo}}$ and the operator $\pdbci$ generally differ; on the strictly positive domain they coincide exactly when $\pfwd\pbwd$ is constant. They are nevertheless interchangeable as decision rules in any setting that depends only on the ordering of outcomes, and in particular on the MAP estimate. The product form is the natural rule when the two boundaries are likelihood-like factors or log-potentials to be jointly imposed; the geometric mean is the natural rule when one seeks a symmetric distributional midpoint under the forward-KL objective in Lemma~\ref{lem:kl-min}. These objectives differ distributionally but agree at the MAP level.
\end{remark}

\section{One Family, Two Readings}
\label{sec:unification}

The arguments of \S\ref{sec:abl}--\S\ref{sec:kl} do not stand in a hierarchy of specification strength. One of them supplies a family and selects no member of it; every other criterion considered here selects, and the criteria fall into two groups according to what they take the two inputs to be. This section states the elementary conditions that force each value, and then places the two derivations of \S\ref{sec:abl} and \S\ref{sec:kl} beside them.

Throughout this section $f$, $b$ and $p$ denote strictly positive weights on $X$ with $|X| \geq 2$, $\mathcal{C}$ denotes normalisation, $u$ the uniform weight, and
\[
W^{*}_{\lambda}[f,b] \;=\; \mathcal{C}\bigl[(f b)^{\lambda}\bigr],
\qquad
U_{L}\,p \;=\; \mathcal{C}\bigl[L\,p\bigr],
\]
written with square brackets to distinguish the map on weights from the distribution $W^{*}_{\lambda}(x)$ of Theorem~\ref{thm:maxent}, so that $W^{*}_{1}[\pfwd,\pbwd] = \pdbci$, and $U_{L}$ is multiplication by a positive weight followed by renormalisation, which is Bayesian updating when $L$ is a likelihood.

\begin{itemize}[noitemsep,topsep=2pt]
  \item \S\ref{sec:maxent} \emph{supplies the family}. Maximum entropy under a moment constraint on $g = \log\pfwd + \log\pbwd$ delivers $W^{*}_{\lambda} \propto (\pfwd\,\pbwd)^{\lambda}$ in its entirety, and selects no member: the multiplier $\lambda$ is fixed only once the moment value is prescribed, which the argument does not supply.
  \item \S\ref{sec:kl} \emph{selects $\lambda = \tfrac{1}{2}$}. The symmetric-KL minimiser is the normalised geometric mean $\sqrt{\pfwd\,\pbwd}/Z'$ of Lemma~\ref{lem:kl-min}, which is $W^{*}_{1/2}$: the distributional midpoint of two things treated as comparable objects.
  \item \S\ref{sec:abl} \emph{selects $\lambda = 1$}. The rank-one reduction of ABL returns the unit-exponent member, with no free exponent parameter once the two boundary weights have been identified as Born weights.
\end{itemize}

\begin{proposition}[Identity and covariance requirements]
\label{prop:exponent-selection}
Let $\lambda > 0$. Then, on strictly positive weights,
\begin{enumerate}[label=\textup{(\roman*)},noitemsep,topsep=2pt]
  \item $W^{*}_{\lambda}[p,p] = \mathcal{C}[p]$ for every $p$ \quad iff \quad $\lambda = \tfrac{1}{2}$;
  \item $W^{*}_{\lambda}[f,u] = \mathcal{C}[f]$ for every $f$ \quad iff \quad $\lambda = 1$;
  \item $W^{*}_{\lambda}[U_{L}f,\, U_{L}b] = U_{L}\,W^{*}_{\lambda}[f,b]$ for every $f, b$ and every positive $L$ \quad iff \quad $\lambda = \tfrac{1}{2}$;
  \item $W^{*}_{\lambda}[U_{L}f,\, b] = U_{L}\,W^{*}_{\lambda}[f,b]$ for every $f, b$ and every positive $L$ \quad iff \quad $\lambda = 1$.
\end{enumerate}
\end{proposition}

\begin{proof}
All four are the same computation. Since normalisation removes positive constants, an identity $\mathcal{C}[A] = \mathcal{C}[B]$ between positive weights holds exactly when $A = cB$ pointwise for some $c > 0$.

For (i), $W^{*}_{\lambda}[p,p] = \mathcal{C}[p^{2\lambda}]$, so the requirement is $p^{2\lambda} = c\,p$. Taking $x, y$ with $p(x) \neq p(y)$ and writing $r = p(x)/p(y) \neq 1$ gives $r^{2\lambda} = r$, hence $2\lambda = 1$. Conversely $\lambda = \tfrac{1}{2}$ returns $\mathcal{C}[p]$. For (ii), $u$ is constant, so $W^{*}_{\lambda}[f,u] = \mathcal{C}[f^{\lambda}]$ and the same argument gives $\lambda = 1$. Each requires only that some non-uniform weight exist, which $|X| \geq 2$ supplies.

For (iii), $W^{*}_{\lambda}[U_{L}f, U_{L}b] = \mathcal{C}\bigl[L^{2\lambda}(fb)^{\lambda}\bigr]$ while $U_{L}W^{*}_{\lambda}[f,b] = \mathcal{C}\bigl[L\,(fb)^{\lambda}\bigr]$, so the requirement is $L^{2\lambda - 1}$ constant. Ranging over non-constant $L$ forces $2\lambda = 1$. For (iv), only one factor carries $L$, giving $L^{\lambda - 1}$ constant and $\lambda = 1$. Both converses are immediate.
\end{proof}

The quantifiers carry the weight. Individual distributions satisfy these equalities accidentally at other exponents: a uniform $p$ satisfies (i) for every $\lambda$, and a constant $L$ satisfies (iii) and (iv) for every $\lambda$. It is therefore the requirement holding across the admissible inputs, not at a chosen one, that identifies the member.

The four conditions divide into two pairs, and the division is the substance of this section. Conditions (i) and (iii) belong to a reading on which the two inputs are equally weighted opinions about the same outcome. On that reading a repeated opinion carries no more than a single one (unanimity preservation), which is (i); and one item of evidence bearing on the shared subject enters both inputs, which is (iii). Both force $\lambda = \tfrac{1}{2}$. Conditions (ii) and (iv) belong to a reading on which the two inputs are separately applied factors whose effects accumulate. On that reading the absence of a second factor leaves the first untouched, which is (ii) and is the requirement recorded in Remark~\ref{rem:neutral}; and evidence acting on one factor passes through the combination unchanged, which is (iv). Both force $\lambda = 1$.

Two of these have names when the inputs are read as probability opinions and $L$ as a likelihood. Condition (iii) is then the external-Bayesianity equation, the requirement that pooling commute with updating every input by a common likelihood; and for the ordinary logarithmic form $\mathcal{C}\bigl[\prod_i p_i^{w_i}\bigr]$ direct substitution into that equation requires $\sum_i w_i = 1$, which for two equally treated inputs is $\tfrac{1}{2}$ each; the axiom is due to Madansky~\cite{Madansky1964}, and its role in characterising pooling operators to Genest, McConway and Schervish~\cite{GenestMcConwaySchervish1986,GenestZidek1986}. The operator of this paper therefore fails external Bayesianity, and it is worth stating that as a fact rather than a matter of jurisdiction: $W^{*}_{1}$ does not satisfy the equation in (iii). What the pair (iii)--(iv) shows is that it is not the only such equation. Condition (iv) is the two-input, one-coordinate form of the individualwise Bayesianity of Dietrich and List~\cite{DietrichList2016}, who introduce that axiom and use it to characterise multiplicative pooling. The two requirements differ in what is updated: external Bayesianity asks that pooling commute with updating \emph{every} input by a \emph{common} likelihood, individualwise Bayesianity that it commute with updating \emph{one} input by a likelihood of its own, the other left untouched. That is why they select different exponents, and the reason is a count: a common update enters the combination through both factors and appears as $L^{2\lambda}$, a single-input update through one factor and appears as $L^{\lambda}$, while the updated combination is $L$ in either case. Dietrich~\cite{Dietrich2010} had previously derived and, for finite hypothesis spaces, characterised the multiplicative class under a distinct Bayesian model of independently supplied information. The distinction between information shared across the inputs and information supplied to them separately, which is what separates the two pairs, is drawn here from Dietrich and List~\cite{DietrichList2016}, who show that the first motivates geometric pooling and the second multiplicative pooling; the shared-versus-private contrast itself is already present in Dietrich~\cite{Dietrich2010}. The rationale we give for it is the double count just described: evidence reaching both inputs is present twice in the product and has to be discounted, which is what the half exponent does, whereas evidence reaching one input alone is present once and should accumulate. Their multiplicative class carries an exogenous calibrating factor, and the bare product is the member obtained when that factor is uniform. Two levels should be kept apart here. In that broader pooling setting the individualwise-Bayesian requirement selects the multiplicative class, and Dietrich and List force the calibrating factor to be uniform with a further axiom, indifference preservation (uniform opinions pool to the uniform one), which they regard as plausible only without shared background information. Here a neutral-boundary requirement does that work, in Remark~\ref{rem:neutral}, which is stated independently of the exponent-selection argument. Within the one-parameter family of this paper the same update equation already forces $\lambda = 1$ on its own, as (iv) shows; the calibration question does not arise, because the family contains no free calibrating factor. We claim neither the distinction nor the characterisation; the factor reading adopted here places \dbci{} on the second side whenever the substrate supplies boundary-local updates, in the sense of \S\ref{sec:operator}.

The two derivations therefore disagree, and they disagree along the same line as the elementary conditions. This is not a defect of either argument: they answer different formal questions, and Proposition~\ref{prop:exponent-selection} shows the difference between their answers to be exactly the difference between the identity and covariance requirements of the two readings.

The disagreement is invisible to the comparison ordinarily made. For every $\lambda > 0$ the map $t \mapsto t^{\lambda}$ is strictly increasing, so all positive-exponent members of the family order the outcomes identically and share an $\arg\max$. A comparison conducted at MAP level therefore cannot distinguish them, and agreement there is a property of the family rather than evidence that two arguments converge. What the exponent governs is not which outcome is selected but how sharply it is preferred, that is, whether the two boundaries compound or average.

That invariance is pointwise, and does not extend to decisions taken after aggregation.

\begin{remark}[Aggregation breaks the invariance]
\label{rem:aggregation}
Proposition~\ref{thm:kl-opt} states that the members of the family agree on the order of individual outcomes. It does not follow that they agree on the order of sums of outcomes. A strictly increasing map preserves the order of its arguments but not the order of their sums, so where the decision is taken between subsets $A, B \subseteq X$ by comparing $\sum_{x \in A} W^{*}_{\lambda}(x)$ against $\sum_{x \in B} W^{*}_{\lambda}(x)$, the exponent can change which subset is preferred while leaving every pairwise comparison of individual outcomes intact.

For a minimal instance take joint weights $\pfwd\pbwd$ proportional to $(0.9,\, 0.01,\, 0.3,\, 0.3)$ with $A = \{x_1, x_2\}$ and $B = \{x_3, x_4\}$. At $\lambda = 1$ the subset scores are $0.91$ and $0.60$, so $A$ is preferred; at $\lambda = \tfrac{1}{2}$ they are $\sqrt{0.9} + \sqrt{0.01} \approx 1.049$ and $2\sqrt{0.3} \approx 1.095$, so $B$ is. The four outcomes stand in the same order throughout.

The two selections of \S\ref{sec:unification} can therefore be separated by a decision rule that scores a class by the mass the operator assigns to it. Where the outcome space is partitioned and the quantity of interest is the class rather than the configuration, as when error configurations are grouped into equivalence classes and the class weights compared, the exponent is not a matter of presentation. It is worth stating in the negative as well: a rule that returns the single highest-weight outcome is unaffected, which is why the question can be overlooked. The disagreement is invisible exactly where comparisons are usually made, and consequential exactly where they are usually acted upon.
\end{remark}

Three kinds of requirement have appeared: an identity property, fixing what it means for an input to add nothing; a covariance property, fixing how new information enters; and an independent derivation. Each appears once on either side, and the split falls in the same place all three times. That is the result of this section, and it is not a tally of six criteria: the members of each column share premises about what the two inputs are, and it is those premises, not the count, that select the exponent.

The operator of this paper is $W^{*}_{1}$, selected by the reading on which the two boundaries are separately applied factors. Rank-one ABL supplies a physical setting that lands on the same member independently, which is why \S\ref{sec:abl} is given the weight it is: not as the sole warrant for the unit exponent, but as a derivation from outside the inference-theoretic premises that agrees with them.

The shared-versus-separate axis has been carried into quantum mechanics before, by a different route. In the quantum state-pooling problem two agents, measuring their own parts of a tripartite system, assign states $\sigma_{A}$ and $\sigma_{B}$ to a third part whose unconditioned state is $\rho$, and a pooler told only those three recovers the state an overseer holding all the data would assign, $\tau \propto \sigma_{A}\rho^{-1}\sigma_{B}$, with $\rho^{-1}$ taken on its support; Spekkens and Wiseman~\cite{SpekkensWiseman2007} identify classes of tripartite state on which that formula holds, and Leifer and Spekkens~\cite{LeiferSpekkens2014} rederive it from a quantum Bayesian conditioning requirement under weaker assumptions. That arrangement, separately acquired data over a commonly held background, is the one under which multiplicative pooling is justified, and the inverse factor discharges the shared background in the role played above by the exogenous calibrating factor; the bare product is again what remains when that background is uniform. Leifer and Spekkens name the multiplicative opinion pool and, for agents who condition a shared prior on separately collected data, obtain the member with unit weights on the two states and weight $-1$ on the prior, which is the bare unit product when that prior is uniform. The correspondence just drawn, between a rank-one post-selection reduction and the separately-supplied side of the split, is not made there. The adjacency is close enough to record.

\begin{remark}[Geometric form]
\label{rem:clr}
On strictly positive weights the family has a compact description. Writing $\operatorname{clr}$ for the centred log-ratio map, $\operatorname{clr}(p)_{x} = \log p(x) - \tfrac{1}{|X|}\sum_{y}\log p(y)$, normalisation is quotienting by constants and
\[
\operatorname{clr} W^{*}_{\lambda}[f,b] \;=\; \lambda\bigl(\operatorname{clr} f + \operatorname{clr} b\bigr).
\]
At $\lambda = 1$ the two log-ratio vectors are added; at $\lambda = \tfrac{1}{2}$ their midpoint is taken. Adding and averaging are the two readings of this section in one line, and the surrounding structure is the Bayes-space geometry of Remark~\ref{rem:product-rules}, in which normalised pointwise multiplication is the group operation and the uniform weight its identity~\cite{vandenBoogaart2010}. The restriction to strict positivity is the one already in force wherever logarithms are used; the operator itself is defined on the wider domain of Remark~\ref{rem:well-defined}.
\end{remark}

\section{Properties of the Operator}
\label{sec:properties}

The properties below concern a tempered family generated by repeated
multiplication by a fixed backward factor. Define
\begin{equation}
\label{eq:T-operator}
T[p](x)=
\frac{p(x)\pbwd(x)}
{\sum_y p(y)\pbwd(y)}.
\end{equation}
A direct calculation gives
\[
T^k[\pfwd](x)
=
\frac{\pfwd(x)\pbwd(x)^k}
{\sum_y\pfwd(y)\pbwd(y)^k}.
\]
Thus $k$ controls the effective strength of the backward boundary;
for $\pbwd(x)\propto e^{-\beta E(x)}$, iteration is exactly annealing
with inverse temperature $k\beta$.

\subsection{Fixed Points of the Iterated Operator}
\label{sec:fixed-point-char}

The fixed points of $T$ admit an exact characterisation, and no parameterisation of $\pbwd$ is needed to state it.

\begin{proposition}[Fixed-Point Characterisation]
\label{prop:fixed-point-char}
Let $\pbwd$ be strictly positive, let $T$ be the iterated \dbci{} operator (\ref{eq:T-operator}), and let $p$ be a probability distribution on $X$. Then
\begin{equation}
\label{eq:fixed-point-char}
T[p] \;=\; p \qquad\Longleftrightarrow\qquad \pbwd \text{ is constant on } \operatorname{supp}(p).
\end{equation}
\end{proposition}

\emph{Proof in Appendix~\ref{app:proofs}.}

\begin{remark}[Structural meaning]
\label{rem:fixed-point-char}
Iteration halts exactly where the backward boundary has stopped discriminating. On a support across which $\pbwd$ takes a single value, every surviving outcome is rescaled by the same constant and the normalised distribution is unchanged; conversely, any distribution left invariant by $T$ must be supported where $\pbwd$ is flat, since the operator would otherwise reweight it. The condition is a statement about $\pbwd$ restricted to $\operatorname{supp}(p)$ and not about its global shape: a backward boundary that discriminates sharply elsewhere in $X$ is inert on such a support. This characterises stationary distributions; \S\ref{sec:fixed-point} establishes convergence, including possible concentration on a single outcome.
\end{remark}

\subsection{Equal-Energy Fixed-Point Structure}
\label{sec:fixed-point}

For the limiting behaviour of the iteration it is convenient to parameterise the backward boundary in exponential-family form,
\begin{equation}
\label{eq:bwd-energy}
\pbwd(x) \;\propto\; e^{-\beta E(x)},
\end{equation}
with energy function $E\colon X \to \mathbb{R}$ and inverse temperature $\beta > 0$, an option made available by the strict positivity of $\pbwd$ (Remark~\ref{rem:well-defined}). The energy is determined by $\pbwd$ only up to an additive constant, but such a shift alters only the proportionality constant in (\ref{eq:bwd-energy}) and hence nothing the operator sees, so no normalisation of $E$ is required. Equal values of $E$ are exactly equal values of $\pbwd$.

\begin{proposition}[Equal-Energy Fixed Point]
\label{prop:fixed-point}
Let $\pfwd$ be a probability distribution on a finite outcome space $X$ and let $\pbwd$ admit the parameterisation (\ref{eq:bwd-energy}). Set
\begin{equation}
E_{\min}^{\pfwd} \;=\; \min\{E(x) : x \in \operatorname{supp}(\pfwd)\}, \qquad M_{\min}^{\pfwd} \;=\; \{x \in \operatorname{supp}(\pfwd) : E(x) = E_{\min}^{\pfwd}\}.
\end{equation}
Then the iterated sequence $p_k = T^{k}[\pfwd]$ converges to
\begin{equation}
\label{eq:fixed-point-limit}
p_\infty(x) \;=\;
\begin{cases}
\dfrac{\pfwd(x)}{\sum_{y \in M_{\min}^{\pfwd}} \pfwd(y)} & x \in M_{\min}^{\pfwd},\\[0.6ex]
0 & \text{otherwise},
\end{cases}
\end{equation}
which is a fixed point of $T$ supported on outcomes of equal (minimum-available) backward energy.
\end{proposition}

\emph{Proof in Appendix~\ref{app:proofs}.}

\begin{remark}[Division of labour at the fixed point]
\label{rem:fixed-point}
At convergence, every surviving outcome carries the same backward energy $E_{\min}^{\pfwd}$, and the distribution of mass within $M_{\min}^{\pfwd}$ is determined by the forward boundary alone. The backward boundary selects the manifold of admissible outcomes; the forward boundary distributes mass within it. The one-shot operator $\pdbci$ takes a single step in this direction; iteration drives it to completion.
\end{remark}

\section{When the Second Boundary Adds No Distinctions Beyond the First}
\label{sec:scope}

The operator takes two arguments, but occupying the second argument is not sufficient to be a second source of case-specific information. This section characterises exactly when the effective backward boundary supplies case-specific information not determined by the forward boundary. The result is elementary. What it does is disqualify constructions that are two-boundary in form but supply only one source of case-specific information; Remark~\ref{rem:scope-limit} states what it does not show.

Throughout, fix a family of admissible cases $\Omega$, and write $P(\omega) = \mathcal{C}[\pfwd]$ for the normalised forward boundary, $B(\omega) = \pbwd$ for the backward weight, and $D(\omega) = \pdbci$ for the resulting operator output associated with a case $\omega \in \Omega$. Here $\mathcal{C}[w](x)=w(x)/\sum_y w(y)$. References to the forward boundary in this section mean this normalised distribution. The effective backward boundary defined below already removes the arbitrary scale of $B$. The informational classification below is always relative to $\Omega$: whether a given boundary varies is a fact about the family, not about the boundary in isolation. Unlike the material of \S\ref{sec:properties}, this section requires only the operator's minimal domain condition $Z > 0$ (Remark~\ref{rem:well-defined}); the strict positivity assumed there for the logarithmic and energy arguments is not used.

\begin{definition}[Effective backward boundary]
\label{def:effective-boundary}
By Remark~\ref{rem:rescaling} the operator depends on $\pbwd$ only up to positive rescaling. For a forward boundary $p$, write
\[
b \sim_p b' \quad\Longleftrightarrow\quad b(x) = c\,b'(x) \ \text{ for all } x \in \operatorname{supp}(p), \ \text{ some } c > 0,
\]
and let $\bar{B}(\omega)$ denote the class of $B(\omega)$ under $\sim_{P(\omega)}$. We call $\bar{B}$ the \emph{effective backward boundary}: it is precisely the part of $\pbwd$ the operator can see. The effective boundary is \emph{forward-determined} if $\bar{B}$ is constant on the fibres of $P$, that is, if any two cases with the same forward boundary have proportional backward weights on its support.
\end{definition}

On strictly positive weights over a fixed support, the algebra underlying the criterion below is standard. Van den Boogaart, Egozcue and Pawlowsky-Glahn~\cite{vandenBoogaart2010} identify densities modulo positive rescaling and define perturbation as their pointwise product, making perturbation by a fixed class an invertible group operation; the cancellation used in the proof is immediate in that setting, and the passage from it to a statement about a family of cases is the elementary factorisation property, that a map is constant on the fibres of another exactly when it factors through it. The quotient has older lineage still in likelihood theory: Fisher defines likelihood only up to proportionality~\cite[p.~310]{Fisher1922}, the likelihood principle treats proportional likelihoods as evidentially equivalent~\cite{Birnbaum1962}, and the converse, that likelihoods posterior-equivalent under every prior are proportional, is Claim~4 of Mayo-Wilson and Saraf~\cite{MayoWilsonSaraf2022}, in the form of their appendix restatement; their main-body statement conjoins support equivalence. What follows extends the formulation to nonnegative weights on the operator's minimal domain $Z > 0$, where supports may vary and a single Bayes space no longer contains all admissible cases, by quotienting the backward weight only on the support visible to the forward boundary. That extension is stratified rather than algebraic: for fixed $p$, the admissible backward weights decompose into ordinary positive Bayes spaces indexed by the zero pattern of $b$ on $\operatorname{supp}(p)$. The statement below is therefore offered not as a new cancellation law but as a family-level criterion for when the second boundary adds nothing the forward boundary does not already determine.

\begin{proposition}[Two-Boundary Reducibility Criterion]
\label{prop:reducibility}
$D$ is determined by $P$ alone on $\Omega$ if and only if the effective backward boundary $\bar{B}$ is forward-determined.
\end{proposition}

\begin{proof}
Let $\omega, \omega'$ satisfy $P(\omega) = P(\omega') = p$, and write $B = B(\omega)$, $B' = B(\omega')$, with normalisers $Z = \sum_y p(y)B(y) > 0$ and $Z' = \sum_y p(y)B'(y) > 0$.

($\Leftarrow$) If $B' = cB$ on $\operatorname{supp}(p)$ with $c > 0$, then $p(x)B'(x) = c\,p(x)B(x)$ for every $x \in \operatorname{supp}(p)$, and both sides vanish off the support; the factor $c$ cancels in the normalisation, so $D(\omega') = D(\omega)$.

($\Rightarrow$) If $D(\omega') = D(\omega)$ then for every $x \in \operatorname{supp}(p)$,
\[
\frac{p(x)B(x)}{Z} \;=\; \frac{p(x)B'(x)}{Z'},
\]
and dividing by $p(x) > 0$ gives $B'(x) = (Z'/Z)\,B(x)$ with $Z'/Z > 0$, so $B \sim_p B'$. No positivity of $B$ or $B'$ is required: the identity holds wherever either weight vanishes.

Hence $D$ is constant on the fibres of $P$ exactly when $\bar{B}$ is.
\end{proof}

\begin{remark}[Why the equivalence class is the right object]
\label{rem:effective-necessary}
Forward determination of $\pbwd$ itself is sufficient for factorisation but not necessary: two cases sharing a forward boundary may carry backward weights $b$ and $7b$, which differ, possibly for external reasons, while the operator output is identical. The rescaling invariance of Remark~\ref{rem:rescaling} is exactly the information the operator discards, and Definition~\ref{def:effective-boundary} quotients by it. Proposition~\ref{prop:reducibility} is therefore a characterisation rather than a sufficient condition, and it is the operator's own invariance that makes it one.
\end{remark}

\begin{remark}[Informational reading]
\label{rem:reducibility-information}
Equip $\Omega$ with any probability law for which the conditional mutual informations below are defined. Given $P$, the operator output $D$ and the effective backward boundary $\bar{B}$ determine one another: $D(x) \propto P(x)B(x)$ and, conversely, $B(x) \propto D(x)/P(x)$ on $\operatorname{supp}(P)$. Hence
\[
I(Y; D \mid P) \;=\; I(Y; \bar{B} \mid P)
\]
for every random variable $Y$ on the case space, so \dbci{} carries information about $Y$ beyond the forward boundary exactly to the extent that its effective backward boundary does. In particular, if $\bar{B}$ is forward-determined then $I(Y; \bar{B} \mid P) = I(Y; D \mid P) = 0$ for every such $Y$.

The information-theoretic content is standard~\cite{CoverThomas}; what is made explicit here for \dbci{} is the identification of $\bar{B}$ as the information quotient this operator induces. In the vocabulary of partial information decomposition~\cite{WilliamsBeer2010}, $I(Y; \bar{B} \mid P)$ is the information unique to $\bar{B}$ together with the synergy between the two boundaries; the consistency conditions of that framework pin the sum to the conditional mutual information, while its division into the two parts depends on the redundancy measure adopted, which Williams and Beer fix as $I_{\min}$ and later proposals vary.

Two directions must be kept apart. For a \emph{fixed} $Y$, $I(Y; \bar{B} \mid P) = 0$ does not imply forward determination: the effective boundary may vary freely and carry nothing relevant to that $Y$. Vanishing for \emph{every} random variable on the case space is a strictly stronger condition and implies it almost surely: taking $Y = \bar{B}$ leaves $\bar{B}$ conditionally determined by $P$ outside a null set of cases. When the law gives every case positive probability, as it can when $\Omega$ is countable, that null set is empty and universal informational irrelevance and forward determination coincide; otherwise the pointwise criterion is the stronger of the two. Single-target irrelevance is strictly weaker still. Conversely, failing to be forward-determined establishes only that an additional input exists, not that it is useful for any particular $Y$; relevance is a target-specific empirical question measured by $I(Y; \bar{B} \mid P)$.
\end{remark}

\begin{remark}[Relation to unique information]
\label{rem:pid}
The nearest condition in the partial-information literature is the vanishing of Bertschinger et al.'s unique-information measure~\cite{BertschingerEtAl2014} for one source relative to another, which holds exactly when the channel from the target to the first is a stochastic garbling of the channel to the second, and so, by Blackwell's theorem~\cite{Blackwell1953}, when the first offers no decision-theoretic advantage over the second. It is not the same condition: under a law giving every case positive probability, Proposition~\ref{prop:reducibility} corresponds to $I(Y; \bar{B} \mid P) = 0$ for every $Y$, which is unique information and synergy together, and is therefore the stronger requirement.
\end{remark}

\begin{remark}[Informational status of backward boundaries]
\label{rem:boundary-classification}
Relative to a family $\Omega$, effective backward boundaries have two statuses: forward-determined or not. The descriptions in \S\ref{sec:operator} and \S\ref{sec:discussion} do not by themselves establish which applies.

\begin{enumerate}[label=\textup{(\roman*)},noitemsep,topsep=2pt]
  \item \emph{Forward-determined.} The effective boundary adds no distinctions between cases sharing the same forward boundary. Two common ways this arises are \emph{forward-derived} weights, computed from $P=\mathcal{C}[\pfwd]$, and \emph{model-fixed} weights, held fixed across $\Omega$, such as a prior or structural constraint. These descriptions may overlap and do not exhaust this status: externally varying positive rescalings also leave the effective boundary unchanged. Bayes with a fixed prior is a model-fixed instance. A fixed boundary can reweight or exclude outcomes and change how outputs differ across cases. It cannot distinguish cases having the same forward boundary.
  \item \emph{Not forward-determined.} Two cases share the same forward boundary but have backward weights that are nonproportional on its support. A per-case post-selection, calibration record, code specification or weak-measurement readout has this status when it produces such a pair.
\end{enumerate}

The distinction concerns $\bar{B}$, not the raw weight $B$ or its provenance. Only the second status adds distinctions between cases sharing the same forward boundary; the informational reading is given in Remark~\ref{rem:reducibility-information}.

The distinction is between \emph{functional reducibility} and \emph{informational provenance}: a model-fixed boundary may encode substantial knowledge whose origin is not the forward channel (an informative prior is not empty) and still leave the rule a function of $\pfwd$, because that knowledge enters once, in the specification, rather than case by case.

Membership is relative to $\Omega$ and not intrinsic to the object. A post-selection state held fixed throughout an experiment is model-fixed. A varying code specification or per-case calibration record is not forward-determined only when cases sharing the same forward boundary have different effective backward boundaries; per-case variation alone does not establish this status. The criterion is silent on \emph{where} a non-forward-determined boundary's content comes from: case-varying noise and an empirical measurement are indistinguishable to it, and which of them obtains is settled by $I(Y; \bar{B} \mid P)$ and by the substrate, not here.
\end{remark}

\begin{remark}[What the criterion can and cannot show]
\label{rem:separating}
The criterion can disqualify a backward boundary as an additional source of case-specific information but cannot in general certify that a boundary is such a source. A boundary computed from the forward boundary, or held fixed across all cases, is forward-determined by construction, and that is visible from how it was specified without reference to any data. Showing that an effective boundary supplies non-forward-determined case-specific information is harder, and on many families it cannot be done with the criterion at all.

The condition is a statement about pairs: whenever two cases share a forward boundary, their effective backward boundaries agree. Testing it means finding such pairs and comparing them. But a universally quantified statement over an empty collection holds automatically, so if distinct cases carry distinct forward boundaries there are no pairs to check, every fibre of $P$ is a singleton, $\bar{B}$ is constant on each for that reason alone, and Proposition~\ref{prop:reducibility} reports that $D$ is determined by $P$. The report is correct, since $P$ then identifies the case and with it $B$ and $D$, but it turns on the uniqueness of the label rather than on anything about the backward boundary. The verdict is exact on this family: the effective second boundary adds no distinctions between cases sharing the same forward boundary. On a sample with no repeated forward boundaries, this does not establish a structural dependence in the underlying population. For the same reason the non-forward-determined class of Remark~\ref{rem:boundary-classification} is empty, since placing a boundary there requires a witnessing pair. As measured forward boundaries seldom repeat exactly, this is the ordinary situation rather than an odd one.

Two routes remain, and neither is free. The conditional mutual information of Remark~\ref{rem:reducibility-information} is the quantity an empirical study would estimate, but it requires a probability law on $\Omega$ under which $P$ does not itself determine the case; on an enumeration of distinct observed cases it vanishes for the same reason the criterion is trivially satisfied. Supplying such a law means modelling the population rather than listing the sample. Alternatively, replace $P$ by a specified summary $\widetilde{P}$ taking values in the probability simplex on $X$ and recompute $\widetilde{D} = \mathcal{C}[\widetilde{P} B]$, with positive overlap. The criterion then applies to the modified operator; merely grouping the original inputs does not preserve its equivalence.
\end{remark}

\begin{remark}[Relation to impossibility results]
\label{rem:impossibility}
Suppose $\bar{B}$ is forward-determined and the induced map $G \colon P \mapsto D$ is admissible in a given rule class (for instance, when the class contains the identity and the map $h$ below and is closed under the \dbci{} product). Then by Proposition~\ref{prop:reducibility} the composite rule lies in the class of rules depending on $\pfwd$ alone, and an impossibility result governing that class applies to it unchanged. Both conditions are needed: admissibility of a map $h$ with $\bar{B} = h \circ P$, whose existence forward determination supplies, does not by itself place $G$ in the class, since multiplication and renormalisation may leave it (a class of linear rules, for instance), and a formally defined but non-admissible map --- a non-computable one, where the result quantifies over computable procedures --- would smuggle in computational power or structure the class excludes.

Two cautions. Membership must be checked against the relevant result's own hypotheses; a rule is not subject to a given impossibility theorem merely by depending on one argument. And the criterion does not establish that a case-varying boundary escapes any particular impossibility, nor that additional information is necessary to do so; relaxing one of the result's hypotheses remains an alternative route. What it establishes is only that a forward-determined boundary is not an informational escape.
\end{remark}

\begin{remark}[What is not claimed]
\label{rem:scope-limit}
Proposition~\ref{prop:reducibility} does not show that a forward channel is insufficient for any target. That $\pfwd$ is insufficient for $Y$ is a substrate-specific claim, established where it holds by the structure of the substrate or by measurement, and it is not established here.

Constancy within a history differs from constancy across histories. For histories $h$, let $P_{t}(h)$ and $B_{t}(h)$ be boundary weights with positive overlap, and set $D_{t}(h) = \mathcal{C}[P_{t}(h)B_{t}(h)]$. At each fixed $t$, Proposition~\ref{prop:reducibility} applies unchanged. A boundary retained within each history may distinguish histories whose current forward states coincide. A boundary shared by all histories adds no such distinctions, yet can continue to reweight an evolving forward state, as the iteration of \S\ref{sec:properties} demonstrates. The criterion tests determination by the current forward boundary; it does not assess the causal benefit of retaining the boundary.
\end{remark}

\begin{remark}[What the criterion is for]
\label{rem:criterion-purpose}
The distinction drawn here has usually been made informally. A second boundary is called second because it is specified separately, or because it refers to a later time, and neither is a statement about information. Proposition~\ref{prop:reducibility} gives an exact test: over the case family, the effective second boundary adds no distinctions between cases sharing the same forward boundary exactly when the output adds none. Remark~\ref{rem:reducibility-information} gives the corresponding informational reading and its conditions. This is what makes the reading offered in \S\ref{sec:introduction} checkable rather than rhetorical. When $\pfwd$ specifies the complete initial state under fixed deterministic dynamics and the backward boundary represents the resulting terminal state, it adds no distinctions between cases sharing the same forward boundary; post-selected quantum mechanics is a physical setting in which it need not be, even when the forward state is completely specified: the maximal description there is the two-state vector, and the post-selection yields information additional to the forward-evolving state~\cite{AharonovVaidman2007}. The criterion applies in any substrate, subject to Remark~\ref{rem:separating}: a construction that computes the backward boundary from the forward one, or holds it fixed across all cases, is ruled out by inspection, while showing that a boundary supplies non-forward-determined case-specific information requires either recurring forward boundaries or the informational test.
\end{remark}

\section{Discussion and Scope}
\label{sec:discussion}

The operator is defined for any pair $(\pfwd, \pbwd)$ of nonnegative boundary weights on a finite outcome space with overlapping support. Its substrate does not enter the formula. When logarithms, KL divergences, or exponential-family parameterisations are used, we additionally assume strict positivity on the relevant outcome space, as in Remark~\ref{rem:well-defined}. The forward boundary \(\pfwd\) might be a squared-amplitude distribution, a likelihood, or a generative-model output; the backward boundary \(\pbwd\) might be the diagonal of a post-selection effect, a code-subspace weighting, a weak-measurement readout represented as a positive likelihood over the outcome space, or a Bayesian prior. What the operator requires is the form of these inputs (two boundary weights on a shared outcome space), not their physical origin. They differ, however, in informational status, which Remark~\ref{rem:boundary-classification} sorts.

The iteration properties of \S\ref{sec:properties} are accordingly consequences that hold in any system meeting that form, whatever its substrate.

\textbf{Born rule.} When $\pbwd$ is uniform, the operator reduces to $\pdbci(x) = \pfwd(x)$. Reading $\pfwd(x) = |\psi(x)|^{2}$ as the squared-amplitude distribution of a pre-selected state $|\psi\rangle$ then recovers the Born rule, the single-boundary case in which the backward boundary carries no information. This is not a derivation of the Born rule --- the squared-amplitude form of $\pfwd$ is itself assumed (it is the Born rule for the forward state) --- but a recovery of it in the limit where $\pbwd$ is uninformative. Deriving the rule rather than assuming it requires premises of a different kind: the decision-theoretic programme of Deutsch and Wallace~\cite{Deutsch1999,Wallace2012} obtains the Born weights from constraints on the preferences of a rational agent, and nothing of that kind is attempted here. The same reduction, obtained there by summing over final states to express indifference to the final condition, is noted by Aharonov, Cohen, Gruss and Landsberger~\cite{AharonovCohenGrussLandsberger2014}.

\textbf{Bayes' rule.} Reading $\pbwd$ as a prior $\pi(x)$ and $\pfwd$ as proportional to the likelihood $p(y \mid x)$ of an observation $y$, the operator gives
\begin{equation*}
\pdbci(x) \;=\; \frac{p(y \mid x)\,\pi(x)}{\sum_{x'} p(y \mid x')\,\pi(x')} \;=\; \pi(x \mid y),
\end{equation*}
which is Bayes' rule; rescaling invariance (Remark~\ref{rem:rescaling}) makes the choice of likelihood normalisation irrelevant. The recovery is interpretive, since the operator's premises do not single out which input is the prior. For the rank-one L\"uders measurement, the forward Born probabilities give $P(n)$, and the backward Born weights give the likelihood $P(F \mid n)$ of the final event $F$. The ABL distribution is therefore $P_{\mathrm{ABL}}(n)=P(n \mid F)$, with the factors supplied by the quantum measurement model.

\textbf{Free-energy view.} Taking the negative logarithm of $\pdbci$,
\begin{equation*}
-\log\pdbci(x) \;=\; \bigl[-\log\pfwd(x)\bigr] + \bigl[-\log\pbwd(x)\bigr] + \log Z,
\end{equation*}
the operator is the Boltzmann distribution of a system whose total energy $E(x) = -\log\pfwd(x) - \log\pbwd(x)$ is the sum of a forward and a backward term, with partition function $Z$ and MAP estimate the minimum-total-energy configuration. This framing makes the exponential-family parameterisation of \S\ref{sec:properties} natural rather than ad hoc: any strictly positive $\pbwd$ is already a Boltzmann distribution in some $E$, and the operator combines two Boltzmann factors into one.

\textbf{Two-state vector formalism.} The rank-one reduction of ABL of \S\ref{sec:abl} is the operator's connection to the two-state vector formalism (TSVF) and to the time-symmetric reading of quantum mechanics developed in subsequent work~\cite{AharonovCohenLandsberger2017}. The TSVF originates with Aharonov, Bergmann and Lebowitz~\cite{AharonovBergmannLebowitz1964} and is reviewed by Aharonov and Vaidman~\cite{AharonovVaidman2007}. In TSVF, ABL is the operational rule for intermediate projective measurements between pre- and post-selected states; the operator is the form that rule takes when the intermediate measurement is rank-one. Weak measurement~\cite{AharonovCohenElitzur2014}, the operational tool most associated with TSVF, is a distinct experimental regime: the present paper concerns the projective-measurement rule on which weak measurement is built, not weak measurement itself. A separate line of work argues that the classical macroscopic regime itself arises from weak measurement over many degrees of freedom together with a final boundary condition, macroscopic averages behaving as deterministic operators in the large-ensemble limit~\cite{CohenAharonov2017}. Whether that route delivers the present operator specifically, as opposed to classicality in general, is open, and we do not rely on it here.

\medskip
\noindent Three limitations concern what the operator alone establishes:

\medskip
\textbf{Joint distributions.} When a joint model is specified, the posterior $P(x \mid y_1, y_2)$ is determined by standard conditioning. The operator reproduces it when its inputs are appropriate factors of that posterior. Multiplying arbitrary conditional distributions can double-count shared information or omit dependencies.

\textbf{Task-specific loss.} The optimal decision depends on the loss function. When $\pdbci$ represents the posterior, it is obtained by minimising posterior expected loss. MAP is optimal for zero--one loss on individual outcomes, while other losses generally require different decisions from the same distribution. The KL minimisation in \S\ref{sec:kl} is a separate distributional problem whose solution is the geometric mean.

\textbf{Latent-variable inference.} The operator is not an alternative to expectation--maximisation or to variational inference. Those procedures iteratively refine an estimate of a latent posterior from data; the operator combines two already-specified boundary distributions. The two problems are different.

Within the stated scope, $\pdbci$ is the unit-exponent member of the maximum-entropy family, selected by the reading of its inputs as separately applied factors and returned independently by the rank-one ABL reduction (\S\ref{sec:unification}); a supplied joint model determines whether this product is the appropriate posterior, while decision losses and latent-inference procedures require further specification.

\section{Conclusion}
\label{sec:conclusion}

\looseness=-1
The \dbci{} operator $\pdbci(x)=\pfwd(x)\pbwd(x)/Z$ is the unit-exponent member of a one-parameter family. Which member a setting calls for is fixed not by the algebra but by what the two inputs are taken to be: separately applied factors, or two equally weighted opinions about the same outcome. Each reading is supported in the same three ways: an identity property, a covariance property and an independent derivation, which on the factor side is the rank-one reduction of ABL. The disagreement between the two readings is invisible to pointwise ranking and to MAP, but not to a rule that scores a class by the mass assigned to it. Under iteration the operator behaves as annealing in the boundary strength: the backward boundary selects a minimum-energy manifold, and the forward boundary distributes mass within it. The Two-Boundary Reducibility Criterion characterises when the effective second argument adds no distinctions between cases sharing the same forward boundary. \dbci{} factors through the forward boundary exactly when the effective backward boundary does. A fixed prior or fixed structural constraint can therefore materially shape every answer without adding such distinctions. Whether a boundary that is not forward-determined is relevant to a given target is a substrate-specific empirical question, and it is there that the two-boundary reading has to show that the questions it suggests lead somewhere.

\section*{Acknowledgements}
\noindent This work was co-funded by the European Union under the Horizon Europe project EuropeanCity\textsuperscript{2} (grant agreement No~101178170). Views and opinions expressed are however those of the author(s) only and do not necessarily reflect those of the European Union or the European Research Executive Agency (REA). Neither the European Union nor the granting authority can be held responsible for them.

\clearpage
\appendix

\begin{center}
{\huge\bfseries Appendices}
\end{center}
\vspace{1.5em}

\section{Proofs for Sections~\ref{sec:maxent}, \ref{sec:kl}, and~\ref{sec:properties}}
\label{app:proofs}

\emph{Proof of Theorem~\ref{thm:maxent}.}
Form the constrained Lagrangian
\begin{equation*}
\mathcal{L}(W) \;=\; -\sum_x W(x)\log W(x) \;+\; \lambda\!\left(\sum_x W(x)\, g(x) - m\right) \;+\; \mu\!\left(\sum_x W(x) - 1\right).
\end{equation*}
Stationarity $\partial\mathcal{L}/\partial W(x) = 0$ yields
\begin{equation*}
-\log W(x) - 1 + \lambda\, g(x) + \mu \;=\; 0,
\end{equation*}
so $W(x) \propto e^{\lambda\, g(x)} = \bigl(\pfwd(x)\,\pbwd(x)\bigr)^{\lambda}$. Normalisation fixes the proportionality constant and gives~(\ref{eq:maxent-form}).

That a suitable $\lambda$ exists, is unique, and ranges over all of $\mathbb{R}$ is established as follows. Assume $g$ non-constant and set
\begin{equation*}
Z(\lambda) \;=\; \sum_x e^{\lambda\, g(x)}, \qquad m(\lambda) \;=\; \mathbb{E}_{W^{*}_{\lambda}}[g] \;=\; \frac{\mathrm{d}}{\mathrm{d}\lambda}\log Z(\lambda),
\end{equation*}
a standard exponential-family computation. Differentiating once more,
\begin{equation*}
\frac{\mathrm{d}m}{\mathrm{d}\lambda} \;=\; \operatorname{Var}_{W^{*}_{\lambda}}(g) \;>\; 0,
\end{equation*}
strictly, because $W^{*}_{\lambda}(x) > 0$ for every $x$ and $g$ is non-constant. Hence $m(\lambda)$ is continuous and strictly increasing. As $\lambda \to +\infty$ the weights $e^{\lambda g(x)}$ concentrate on $\arg\max_x g$, so $m(\lambda) \to \max_x g(x)$; symmetrically $m(\lambda) \to \min_x g(x)$ as $\lambda \to -\infty$. Therefore $\lambda \mapsto m(\lambda)$ is a strictly increasing bijection from $\mathbb{R}$ onto the open interval $(\min_x g, \max_x g)$, which is exactly the interior of the convex hull of $\{g(x)\}$. Every interior $m$ thus determines a unique finite $\lambda$.

This also closes the maximisation over the whole feasible set. The $W^{*}_{\lambda}$ so obtained is strictly positive and feasible, and satisfies the stationarity conditions above; since $H$ is strictly concave and the feasible set (the simplex intersected with the moment hyperplane) is convex, a feasible stationary point of a strictly concave function is its unique global maximiser, boundary points of the simplex included. Substituting $\lambda = 1$ into~(\ref{eq:maxent-form}) gives the multiplicative form $\pfwd(x)\,\pbwd(x)/Z$, identifying $W^{*}_{1}$ with $\pdbci$.

The argument above covers the interior case, in which the moment constraint is met at a finite $\lambda$. At the extreme feasible value $m = \max_x g(x)$ it is not: since $g(x) \le m$ for every $x$, the constraint $\sum_x W(x)\,g(x) = m$ forces $W$ to vanish off the level set $L_{\max} = \{x : g(x) = m\}$, and every distribution supported on $L_{\max}$ is feasible. Entropy on that face of the simplex is maximised by the uniform distribution on $L_{\max}$, which is $\lim_{\lambda \to +\infty} W^{*}_{\lambda}$, since $W^{*}_{\lambda}$ places weight $\propto e^{\lambda g(x)}$ and the ratio of any off-$L_{\max}$ weight to an on-$L_{\max}$ weight tends to zero. The case $m = \min_x g(x)$ is symmetric, with $\lambda \to -\infty$.

\medskip
\emph{Proof of Lemma~\ref{lem:kl-min}.}
Expanding (\ref{eq:kl-sym}) and collecting the two logarithms,
\begin{equation*}
D_{\mathrm{sym}}(p')
\;=\; \sum_x p'(x)\ln p'(x) \;-\; \tfrac{1}{2}\sum_x p'(x)\ln\bigl(\pfwd(x)\,\pbwd(x)\bigr)
\;=\; \sum_x p'(x)\,\ln\frac{p'(x)}{\sqrt{\pfwd(x)\,\pbwd(x)}}.
\end{equation*}
Writing $C = \sum_y \sqrt{\pfwd(y)\,\pbwd(y)}$, so that $\sqrt{\pfwd(x)\,\pbwd(x)} = C\,p'_{\mathrm{geo}}(x)$ with $p'_{\mathrm{geo}}$ as in (\ref{eq:kl-min}), this is the identity
\begin{equation*}
D_{\mathrm{sym}}(p') \;=\; D_{\mathrm{KL}}\bigl(p' \,\|\, p'_{\mathrm{geo}}\bigr) \;-\; \ln C .
\end{equation*}
The second term does not depend on $p'$, so minimising $D_{\mathrm{sym}}$ is minimising $D_{\mathrm{KL}}(p' \| p'_{\mathrm{geo}})$, which is non-negative and vanishes precisely at $p' = p'_{\mathrm{geo}}$. Since $\pfwd$ and $\pbwd$ are strictly positive, $p'_{\mathrm{geo}}$ is strictly positive and the divergence is finite for every $p'$; the argument therefore holds over the whole simplex, boundary included, and requires no differentiation.

\medskip
\emph{Proof of Proposition~\ref{thm:kl-opt}.}
For $\lambda > 0$ the map $t \mapsto t^{\lambda}$ is strictly increasing on $[0,\infty)$, so $W^{*}_{\lambda}(x) \propto (\pfwd(x)\,\pbwd(x))^{\lambda}$ is a strictly increasing function of $\pfwd(x)\,\pbwd(x)$, and normalisation by the positive constant $Z_{\lambda}$ preserves order. Every member of the family therefore induces the same ordering of $X$ as the product itself, and in particular the same maximiser set, with ties on one side ties on the other. Taking $\lambda = 1$ and $\lambda = \tfrac{1}{2}$ gives the stated case for $\pdbci$ and, by Lemma~\ref{lem:kl-min}, for $p'_{\mathrm{geo}}$.

\medskip
\emph{Proof of Proposition~\ref{prop:fixed-point-char}.}
Write $Z_p = \sum_y p(y)\,\pbwd(y)$, which is strictly positive since $\pbwd > 0$, so $T[p]$ is defined.

($\Leftarrow$) Suppose $\pbwd(x) = c$ for some constant $c > 0$ at every $x \in \operatorname{supp}(p)$. Then $Z_p = c$ by normalisation of $p$, and substituting into (\ref{eq:T-operator}) gives
\begin{equation*}
T[p](x) \;=\; \frac{p(x)\,c}{c} \;=\; p(x)
\end{equation*}
at every $x \in \operatorname{supp}(p)$, while both sides vanish elsewhere.

($\Rightarrow$) Suppose $T[p] = p$. For $x \in \operatorname{supp}(p)$ this reads $p(x)\,\pbwd(x)/Z_p = p(x)$, and dividing by $p(x) > 0$ gives $\pbwd(x) = Z_p$. The value $Z_p$ does not depend on $x$, so $\pbwd$ is constant on $\operatorname{supp}(p)$.

\medskip
\emph{Proof of Proposition~\ref{prop:fixed-point}.}
A direct calculation gives $p_k(x) \propto \pfwd(x)\,\pbwd(x)^{k} \propto \pfwd(x)\,e^{-k\beta E(x)}$. For $x \in M_{\min}^{\pfwd}$, the factor $e^{-k\beta E_{\min}^{\pfwd}}$ is common and cancels under normalisation; for $x \in \operatorname{supp}(\pfwd) \setminus M_{\min}^{\pfwd}$, the ratio $e^{-k\beta(E(x) - E_{\min}^{\pfwd})}$ vanishes as $k \to \infty$ since $E(x) > E_{\min}^{\pfwd}$ and $\beta > 0$. The limiting distribution is therefore supported on $M_{\min}^{\pfwd}$ with weights proportional to $\pfwd$, giving (\ref{eq:fixed-point-limit}). Fixed-point invariance is then immediate: every $x \in M_{\min}^{\pfwd} = \operatorname{supp}(p_\infty)$ carries the same energy $E_{\min}^{\pfwd}$, so $\pbwd$ is constant on $\operatorname{supp}(p_\infty)$, and Proposition~\ref{prop:fixed-point-char} gives $T[p_\infty] = p_\infty$.

\bibliographystyle{unsrt}
\bibliography{references}

\end{document}